\pdfoutput=1
\documentclass[11pt]{article}
\usepackage[T1]{fontenc}
\usepackage[latin9]{inputenc}
\usepackage{amsfonts}
\usepackage{amssymb}
\usepackage{amsthm}
\usepackage{amsmath}
\usepackage{comment}
\usepackage[verbose,margin=1in]{geometry}
\usepackage{enumerate}
\usepackage{microtype}
\usepackage{algorithm}
\usepackage{algorithmicx}
\usepackage{algpseudocode}
\usepackage{mathtools}
\usepackage{thm-restate}
\usepackage{mathrsfs}
\usepackage{csquotes}
\usepackage{apptools}
\usepackage{graphicx}
\usepackage{verbatim}
\usepackage{float}
\usepackage{chngcntr}
\usepackage{xr}
\usepackage{bm}
\usepackage[disable]{todonotes}
\usepackage[shortlabels]{enumitem}
\usepackage{pdfsync}
\usepackage{nicefrac}

\usepackage[hypertexnames=false,backref=page]{hyperref}
\usepackage[capitalize, nameinlink]{cleveref}

\makeatletter 

\numberwithin{equation}{section}
\numberwithin{figure}{section}

\theoremstyle{plain}
\newtheorem{theorem}{Theorem}

\newtheorem{corollary}[theorem]{Corollary}
\newtheorem{lemma}[theorem]{Lemma}

\newtheorem{claim}[theorem]{Claim}

\theoremstyle{remark}

\theoremstyle{definition}
\newtheorem{definition}[theorem]{Definition}

\algnewcommand{\LeftComment}[1]{\(\triangleright\) #1}

\global\long\def\defeq{\stackrel{\mathrm{{\scriptscriptstyle def}}}{=}}

\def\eps{\varepsilon}
\global\long\def\R{\mathbb{R}}

\global\long\def\rank{\mathrm{rank}}

\global\long\def\exp{\mathrm{exp}}

\global\long\def\E{\mathbb{E}}

\DeclarePairedDelimiter\inner{\langle}{\rangle}
\DeclarePairedDelimiterX\norm[1]\lVert\rVert{\ifblank{#1}{\blank}{#1}}

\global\long\def\nnr{\mathrm{rank}_+}

\global\long\def\xc{\mathrm{xc}}

\global\long\def\cD{\mathcal{D}}
\global\long\def\cE{\mathcal{E}}
\global\long\def\cF{\mathcal{F}}

\global\long\def\cK{\mathcal{K}}
\global\long\def\cP{\mathcal{P}}
\global\long\def\cQ{\mathcal{Q}}
\global\long\def\cV{\mathcal{V}}

\global\long\def\vzero{\bm{0}}
\global\long\def\vone{\bm{1}}

\global\long\def\mzero{\bm{0}}
\global\long\def\ma{\mathbf{A}}
\global\long\def\mb{\mathbf{B}}
\global\long\def\mc{\mathbf{C}}

\global\long\def\mm{\mathbf{M}}
\global\long\def\mr{\mathbf{R}}
\global\long\def\ms{\mathbf{S}}
\global\long\def\mt{\mathbf{T}}
\global\long\def\mmu{\mathbf{U}}
\global\long\def\mv{\mathbf{V}}

\global\long\def\vb{\bm{b}}

\global\long\def\ve{\bm{e}}

\global\long\def\vg{\bm{g}}

\global\long\def\vu{\bm{u}}
\global\long\def\vv{\bm{v}}

\global\long\def\vx{\bm{x}}
\global\long\def\vy{\bm{y}}

\title{Polytopes with Bounded Integral Slack Matrices Have Sub-Exponential Extension Complexity}
\author{
	Sally Dong\thanks{University of Washington. \texttt{sallyqd@uw.edu.}} \and
	Thomas Rothvoss\thanks{University of Washington. \texttt{rothvoss@uw.edu.} Supported by NSF CAREER grant 1651861, NSF grant 2318620 and a David \& Lucile Packard Foundation Fellowship.}
}

\begin{document}
\maketitle

\begin{abstract}
We show that any bounded integral function $f : A \times B \mapsto \{0,1, \dots, \Delta\}$ with rank $r$ 
has deterministic communication complexity $\Delta^{O(\Delta)} \cdot \sqrt{r} \cdot \log r$, 
where the rank of $f$ is defined to be the rank of the $A \times B$ matrix whose entries are the function values.
As a corollary, we show that any $n$-dimensional polytope that admits a slack matrix with entries from $\{0,1,\dots,\Delta\}$ has extension complexity at most $\exp(\Delta^{O(\Delta)} \cdot \sqrt{n} \cdot \log n)$.
	
\end{abstract}

\section{Introduction} \label{sec:intro}

In classical communication complexity, two players, Alice and Bob, are given a Boolean function $f : A \times B \mapsto \{0,1\}$, as well as separate inputs $a \in A$ and $b \in B$,
and wish to compute $f(a,b)$ while minimizing the total amount of communication. 
Alice and Bob have unlimited resources for pre-computations and agree on a deterministic \emph{communication protocol} to compute $f$ before receiving their respective inputs.
The \emph{length} of their protocol is defined to be the maximum number of bits exchanged over all possible inputs.
The \emph{deterministic communication complexity} of $f$, denoted by $CC^{\det}(f)$,
is the minimum length of a protocol to compute $f$.

A major open problem in communication complexity is the \emph{log-rank conjecture} proposed by Lov{\'a}sz and Saks \cite{lovasz1988lattices}, which asks if $CC^{\det}(f) \leq (\log {r})^{O(1)}$ for all Boolean functions $f$ of rank $r$,
where the \emph{rank} of a two-party function $f$ on $A \times B$ is defined to be the rank of the matrix $\mm \in \R^{A \times B}$ with $\mm_{a,b} = f(a,b)$ for all $(a, b) \in A \times B$.
The best upper bound currently known is due to Lovett \cite{lovett2016}, who showed $CC^{\det}(f) \leq O(\sqrt{r} \log r)$ using discrepancy theory techniques.

In this work, we obtain similar deterministic communication complexity bounds for a larger class of functions:

\begin{theorem}[Main result, communication complexity] \label{thm:main}
	Let $f : A \times B \mapsto \{0,  1, \dots, \Delta\}$ be a bounded integral function of rank $r$.
	Then there exists a deterministic communication protocol to compute $f$ with length at most $\Delta^{O(\Delta)} \cdot \sqrt{r} \cdot \log r$ bits.
\end{theorem}

The function $f$ can be directly viewed as a non-negative matrix $\mm \in \{0, \dots, \Delta\}^{A \times B}$ of rank $r$.
We use the matrix representation exclusively in the remainder of this paper.
Let us adopt the convention that a \emph{rectangle} in $\mm$ is a (non-contiguous) submatrix $\mm[A',B']$ indexed by some $A' \subseteq A$ and $B' \subseteq B$. A rectangle is \emph{monochromatic} with color $i$ if all entries in the rectangle have value $i$.

The \emph{non-negative rank} of a non-negative matrix $\mm$, denoted by $\nnr(\mm)$, is defined as the minimum $r$ 
such that $\mm$ can be written as the sum of $r$ non-negative rank-1 matrices,
or equivalently, as $\mm = \mmu \mv$ for non-negative matrices $\mmu \in \R^{A \times r}_{\geq 0}$ and $\mv \in \R^{r \times B}_{\geq 0}$.
It is straightforward to see $\nnr(\mm) \leq 2^{CC^{\det}(\mm)}$ (c.f. Rao and Yehudayoff \cite{rao2020communication}, Chapter 1, Theorem 1.6):
The \emph{protocol tree} to compute $\mm$ has at most $2^{CC^{\det}(\mm)}$ leaves,
each corresponding to a monochromatic rectangle of $\mm$. 
These rectangles are disjoint over all leaves, and their union is $\mm$. 
Since a monochromatic rectangle is a non-negative matrix of rank 0 or 1, 
we conclude that $\mm$ can be written as a sum of at most $2^{CC^{\det}(\mm)}$ rank-1 non-negative matrices.
The \emph{positive semidefinite rank} of $\mm$, denoted by $\rank_{\mathrm{psd}}(\mm)$, generalizes non-negative rank and has an analogous relationship to quantum communication complexity~\cite{fawzi2015positive}. It is defined as the minimum $r$ such that there are positive semidefinite matrices $\mmu_1, \dots \mmu_A$ and $\mv_1,\dots, \mv_B$ of dimension $r \times r$ satisfying $\mm_{i,j} = \mathrm{Tr}(\mmu_i \mv_j)$ for all $i, j$; trivially, $\rank_{\mathrm{psd}}(\mm) \leq \rank_+(\mm)$. Barvinok~\cite{barvinok2012approximations} showed that if $\mm$ has at most $k$ distinct entries, as is the setting studied in this paper, then $\rank_{\mathrm{psd}}(\mm) \leq {k-1 + \rank(\mm) \choose k-1}$.

Non-negative rank brings us to a beautiful connection with \emph{extension complexity},
introduced in the seminal work of Yannakakis \cite{yannakakis1988} in the context of writing combinatorial optimization problems as linear programs.
The \emph{extension complexity} of a polytope $\cP$, denoted $\xc(\cP)$, is defined as the minimum number of facets of some higher dimensional polytope $\cQ$ (its \emph{extended formulation}) such that there exists a linear projection of $\cQ$ to $\cP$.

A foundational theorem from \cite{yannakakis1988} states that $\xc(\cP) = \nnr(\ms)$, where $\ms$ is the \emph{slack matrix} of $\cP$, defined as follows:
Suppose $\cP$ has facets $\cF$ and vertices $\cV$. Then $\ms$ is a non-negative $\cF \times \cV$ matrix, where the $(f,v)$-entry indexed by facet $f \in \cF$ defined by the halfspace $a^\top x \leq b$ and vertex $v \in \cV$ has value $\ms_{f,v} = b - a^\top v$.
(A facet may be defined by many equivalent halfspaces, and therefore the slack matrix is not unique; the result holds for all valid slack matrices.)
In fact, \cite{yannakakis1988} showed that a factorization of $\ms$ with respect to non-negative rank gives an extended formulation of $\cP$ and vice versa. For a comprehensive preliminary survey, see \cite{conforti2010extended}.
A number of breakthrough results in extended complexity in recent years emerged from lower-bounding the non-negative rank of the slack matrix for specific polytopes, such as the \textsc{TSP}, \textsc{Cut}, and \textsc{Stable-Set} polytopes by Fiorini-Massar-Pokutta-Tiwary-De Wolf \cite{fiorini2015exponential}, and the \textsc{Perfect-Matching} polytope by Rothvoss \cite{rothvoss2017matching}.

Connecting extension complexity to deterministic communication complexity via non-negative rank, we have:

\begin{corollary}[Main result, extension complexity] \label{cor:main-xc}
	Let $\cP$ be a $n$-dimensional polytope that admits slack matrix $\ms$. 
	Suppose the entries of $\ms$ are integral and bounded by $\Delta$.
	Then the extension complexity of $\cP$ is at most $\exp( \Delta^{O(\Delta)} \cdot \sqrt{n} \cdot \log n)$.
\end{corollary}
\begin{proof}
	Since $\cP$ is $n$-dimensional, we know $\ms$ has rank at most $n$, as there are at most $n$ linearly independent vertices of $\cP$, and the slack of a vertex with respect to all the facets is a linear function.
	We combine the theorem of \cite{yannakakis1988} and \cref{thm:main} for the overall conclusion.
\end{proof}

Our proof follows the approach discussed in the note of Rothvoss  \cite{rothvoss2014direct} simplifying Lovett's result. 
The following lemma shows that there are indeed concrete polytopes for which our result applies:

\begin{lemma} \label{lem:sandwich}
	Suppose there are polytopes $\cP \subseteq \cQ \subseteq \R^n$ where $\cP = \mathrm{conv}\{\vx_1, \dots, \vx_v\}$
	and $\cQ = \{ \vx \in \R^n : \ma \vx \leq \vb\}$ with $\ma \in \R^{f \times n}$.
	Suppose the partial slack matrix $\ms \in \R^{f \times v}$ with $\ms_{i,j} = \vb_i - \ma_i \vx_j$ for $i \in [f], j \in [v]$ is integral and bounded by $\Delta$.
	Then there exists a polytope $\cK$ with extension complexity at most $\exp(\Delta^{O(\Delta)} \cdot \sqrt{n} \cdot \log n)$ so that $\cP \subseteq \cK \subseteq \cQ$. 
\end{lemma}
\begin{proof}
	From above, we have $\nnr(\ms) \leq \exp(\Delta^{O(\Delta)} \cdot \sqrt{n} \cdot \log n)$.
	Moreover, it is well-known that $s \defeq \nnr(\ms)$ is the extension complexity of some polytope $\cK$ such that $\cP \subseteq \cK \subseteq \cQ$. (It is in fact the minimum extension complexity over all such sandwiched polytopes.) The conclusion follows.
	
	For completeness, we show the latter fact: 
	Suppose $\ms = \mmu \mv$ is a non-negative factorization of $\ms$ with $\mmu \in \R_{\geq 0}^{f \times s}$ and $\mv \in \R^{s \times v}_{\geq 0}$. Let $\cK^{\mathrm{lift}} \defeq \{ (\vx, \vy) \in \R^{n + s} : \ma \vx + \mmu \vy = \vb, \vy \geq \vzero\}$, and let $\cK$ be the projection of $\cK^{\mathrm{lift}}$ onto the first $n$ coordinates. 
	It is immediately clear that $\cK \subseteq \cQ$, and $\xc(\cK) \leq s$ by definition.
	For each $j \in [v]$, the point $\vx_j$ satisfies $\ma \vx_j + \ms^j = \vb$, where $\ms^j$ is the $j$-th column of $\ms$ given by $\mmu \mv \ve_j$. It follows that $(\vx_j, \mv \ve_j) \in \cK^{\mathrm{lift}}$, so $\vx_j \in \cK$. As this holds for each $\vx_1, \dots, \vx_v$, we conclude $\cP \subseteq \cK$. 
	
\end{proof}

We give a direct example in a combinatorial optimization setting:
Consider the \textsc{$k$-Set-Packing} problem, where we are given a collection of $n$ sets $S_1, \dots, S_n \subseteq [N]$ with $N \gg n$,
and want to find a maximum subcollection such that each element $j \in [N]$ is contained in at most $k$ sets. 
The \textsc{$k$-Set-Packing} polytope $\cP$ is is the convex hull of all feasible subcollections of sets, given by
\[
	\cP = \mathrm{conv}\Big\{ x \in \{0,1\}^n : \sum_{i : j \in S_i} x_i \leq k \;\; \forall j \in [N]\Big\}.
\]
Its natural LP relaxation $\cQ$ is
\[
	\cQ = \Big\{ x \in [0,1]^n : \sum_{i : j \in S_i} x_i \leq k \;\; \forall j \in [N]\Big\}.
\]
In the regime where $N \gg n$, a priori, the extension complexity of $\cP$ and $\cQ$ could be as large as $N$. 
But interestingly, let $\ms$ be the partial slack matrix with respect to $\cP$ and $\cQ$ as defined in \cref{lem:sandwich}.
Then $\ms$ contains integral values in $\{0, \dots, k\}$, and so we conclude there exists a sandwiched polytope $\cP  \subseteq \cK \subseteq \cQ$ with $\xc(\cK) \leq \exp(k^{O(k)} \cdot \sqrt{n} \cdot \log n)$. 

\section{Communication protocol} \label{sec:comm-protocol}

In this section, we give a deterministic communication protocol for a bounded integral matrix $\ms$ assuming it has the crucial property that any submatrix contains a large monochromatic rectangle. The protocol is based on the protocol from Nisan and Widgerson~\cite{nisan1995rank} that is expanded on by Lovett~\cite{lovett2016}.

\begin{lemma}\label{lem:comm-protocol}
Let $0 < \delta < 1$.
Let $\mm \in \{0, 1, \dots, \Delta\}^{A \times B}$ be a bounded integral matrix with rank $r$,
and suppose for any submatrix $\ms \defeq \mm[A', B']$ where $A' \subseteq A$ and $B' \subseteq B$,
there is a monochromatic rectangle in $\ms$ of size $\geq \exp(-\delta(r)) |A'||B'|$ for some function $\delta$ of $r$.
Then,
\[
	CC^{\det}(\mm) \leq \Theta \left(\log \Delta +  \log^2 r  + \sum_{i=0}^{\log r} \delta(r/2^i) \right).
\]
\end{lemma}

We begin by proving two helper lemmas relating to $\mm$ and the rank of its submatrices, which we subsequently use to bound the communication complexity.

\begin{lemma} \label{lem:rank-reduce}
	Let $\ma, \mb, \mc, \mr$ be matrices of the appropriate dimensions,
	and let $\mr$ have rank 0 or 1. Then
	\[
		\rank \begin{pmatrix}
			\mr \\ \mb 
		\end{pmatrix} +
		\rank \begin{pmatrix}
			\mr & \ma
		\end{pmatrix} \leq 
		\rank \begin{pmatrix}
			\mr & \ma \\
			\mb & \mc
		\end{pmatrix} + 3.
	\]
\end{lemma}
\begin{proof}
	We use a sequence of elementary rank properties:
	\begin{align*}
	\rank \begin{pmatrix}
		\mr & \ma \\
		\mb & \mc
	\end{pmatrix} +1 \geq 
	\rank \begin{pmatrix}
		\mzero & \ma \\
		\mb & \mc
		\end{pmatrix} &\geq 
	\rank (\ma) + \rank (\mb) \\
	& \geq 
	\rank \begin{pmatrix}
		\mr & \ma
	\end{pmatrix} +
	\rank \begin{pmatrix}
		\mr \\ \mb 
		\end{pmatrix} - 2.
	\end{align*}
\end{proof}

\begin{lemma} \label{lem:size-reduce}
	Suppose the rank of $\mm \in \{0, 1, \dots, \Delta\}^{A \times B}$ is $r$.
	Then $\mm$ contains at most $(\Delta + 1)^r$ different rows and columns.
\end{lemma}

\begin{proof}
	We show the argument for rows:
	Let $\mmu \in \{0,\dots,\Delta\}^{A \times r}$ denote a submatrix of $\mm$ consisting of $r$ linearly independent columns of $\mm$. 
	Clearly $\mmu$ has at most $(\Delta + 1)^r$ different rows. If row $i$ and row $j$ of $\mmu$ are identical, then row $i$ and $j$ of $\mm$ are also identical, since by definition of rank, the columns of $\mm$ are obtained by taking linear combinations of columns of $\mmu$.
\end{proof}

Now, we can design a communication protocol to compute $\mm$ using standard techniques.
Without loss of generality, we may assume $\mm$ does not contain identical rows or columns, so that the conclusion of \cref{lem:size-reduce} can be applied.

\begin{proof}[of \cref{lem:comm-protocol}]
	Suppose Alice has input $a \in A$ and Bob has $b \in B$.
	Then Alice and Bob compute $\mm_{a,b}$ by recursively reducing the matrix $\mm$ to a smaller submatrix in one of two ways: 
	they communicate the bit 0 which guarantees a decrease in rank in the resulting submatrix, and the bit 1 which guarantees a decrease in size.
	
	Let $\ms$ denote the submatrix to be considered at a recursive iteration.
	Alice and Bob begin with $\ms \defeq \mm$.
	During a recursive iteration, they first write $\ms$ in the form
		\[
			\ms = \begin{pmatrix}
							\mr & \ma \\
							\mb & \mc \\
			\end{pmatrix},
		\]
		where $\mr$ denote the large monochromatic rectangle in $\ms$ that is guaranteed to exist, with
		$|\mr| \geq \exp(-\delta(r)) |\ms|$.
		By \cref{lem:rank-reduce}, either $\rank \begin{pmatrix} \mr & \ma \end{pmatrix} \leq \frac{1}{2} \rank (\ms) + \frac32$,
		or $\rank \begin{pmatrix} \mr & \mb \end{pmatrix} \leq \frac{1}{2} \rank (\ms) + \frac32$.
		
		In the first case, Alice communicates. 
		If her input row $a$ is in the upper submatrix of $\ms$, Alice sends the bit 0, and both Alice and Bob update $\ms =  \begin{pmatrix} \mr & \ma \end{pmatrix}$.
		On the other hand, if row $a$ is in the lower submatrix of $\ms$, Alice sends the bit 1, and they both update
		$\ms = \begin{pmatrix} \mb & \mc \end{pmatrix}$.
		In the second case, Bob communicates. If his input column $b$ is in the left submatrix of $\ms$, Bob sends 0, and they update $\ms = \begin{pmatrix} \mr \\ \mb \end{pmatrix}$.
		Otherwise, Bob sends 1, and they update $\ms = \begin{pmatrix} \ma \\ \mc \end{pmatrix}$.
		If $\ms$ has size 1, Alice can simply output the entry of $\ms$, which is precisely $\mm_{a,b}$.
	
		Consider the communication protocol up until the rank of $\ms$ is halved (i.e., when the first 0 bit is communicated):
		The protocol tree has at most $\Theta \left(\frac{- \log |\ms|}{\log(1 - \exp(-\delta(r)))}\right) = \Theta \left( \frac{r \log(\Delta+1)}{\exp(-\delta(r))} \right)$-many leaves at this point, which is the max number of 1 bits that could have been communicated.
		Standard balancing techniques (c.f.\cite[Chapter 1, Theorem 1.7]{rao2020communication}) then allow us to balance the protocol tree; combined with the bound $|\ms| \leq (\Delta+1)^{2r}$ by \cref{lem:size-reduce}, 
		we conclude that there exists a protocol of length $O(\log r + \log \log(\Delta+1) + \delta(r))$.
		
		Next, consider the phase where the protocol continues until the rank drops from $r/2$ to $r/4$.
		Note that we may assume the submatrix at the start of this phase has unique rows and columns, as Alice and Bob has unlimited computation at the start. Then this phase can be simulated by a protocol of length $O(\log (r/2) + \log \log (\Delta+1) + \delta(r/2))$. We proceed in the same fashion where at phase $i$, the rank drops from $r/2^i$ to $r/2^{i+1}$. Summing over $i=0, \dots, \log r$, we get that the total protocol length is at most
		\begin{align*}
			&\phantom{{}={}} \sum_{i=0}^{\log r} \left( \log\Big(\frac{r}{2^i}\Big) + \log \log (\Delta + 1) + \delta\Big(\frac{r}{2^i}\Big)\right) \\
			&\leq \log^2 r + \log r \log \log (\Delta+1) + \sum_{i=0}^{\log r} \delta\Big(\frac{r}{2^i}\Big).
		\end{align*}

		In the base case, $\ms$ has rank 1 or 2. Suppose it has rank 2, and let $\ms = \vu \vv^\top + \vu' \vv'^\top$ be a factorization,
		where $\vu$ and $\vu'$ are two linearly independent columns of $\ms$.
		Then Alice sends $\vu_a$ and $\vu'_a$, which Bob uses to compute $\vu_a \vv_b + \vu'_a \vv'_b = \ms_{a,b} = \mm_{a,b}$. Since $\vu$ and $\vu'$ are columns of $\ms$, their entries take values from $\{0,1,\dots, \Delta\}$, so Alice communicates $O(\log \Delta)$ bits in total.
                Finally we can omit the term $\log r \log \log (\Delta+1)$ as it is dominated by $\log^2 r + \log \Delta$.
\end{proof}

\section{Finding large monochromatic rectangles} \label{sec:monochrom-rect}

In this section, we show that the assumption for applying the communication protocol from \cref{sec:comm-protocol} does indeed hold. That is, any bounded integral matrix contains a sufficiently large monochromatic rectangle.

We first reduce the problem of finding large monochromatic rectangles to finding \emph{almost-monochromatic} rectangles. 
We say a rectangle is \emph{$(1-\eps)$-monochromatic} if at least a $(1-\eps)$-fraction of its entries have the same value.

\begin{lemma} \label{lem:almost-monochrom-rec-reduction}
	Suppose $\mm \in \R^{A \times B}$ has rank $r \geq 1$, and is $(1 - \frac{1}{16r})$-monochromatic with color $\alpha$.
	Then $\mm$ contains a monochromatic rectangle of size $\geq \frac{|A||B|}{8}$.
\end{lemma}
\begin{proof}
	Let us call a column of $\mm$ \emph{bad} if it contains at least $\frac{1}{8r} |A|$-many non-$\alpha$ entries.
	By Markov's inequality, at most half the columns are bad.
	Let $B' \subseteq B$ be the remaining good columns, where each contains at least $(1-\frac{1}{8r}) |A|$-many $\alpha$ entries.
	Let $B'' \subseteq B'$ be a maximal set of linearly independent good columns. 
	We know $|B''| \leq r$ as the rank of $\mm$ is $r$.
	
	Let $\mmu = \mm[A,B'']$. 
	Since each column in $B''$ contains at most $\frac{1}{8r} |A|$-many non-$\alpha$ entries, 
	there are at most $r \cdot \frac{1}{8r} |A|$ rows of $\mmu$ that contain non-$\alpha$ entries.
	Let $A'$ denote the $|A| - \frac{1}{8}|A| \geq \frac{1}{2} |A|$ rows of $\mmu$ that contain only $\alpha$ entries.
	So $\mm[A',B'']$ contains only $\alpha$ entries.
	
	Let $\mt = \mm[A',B']$. Since the columns in $B'$ are linear combinations of columns in $B''$, each column of $\mt$ must be of the form $\beta \vone$ for some $\beta$. 
	Finally, we know $|\mt| =  |A'||B'| \geq \frac{1}{4} |A||B|$, and there are at most $\frac{1}{16r}|A||B|$ non-$\alpha$ entries in $\mm$ in total, so at least half the columns of $\mt$ must have value $\alpha$.
\end{proof}

Now, it remains to show that we can find large almost-monochromatic rectangles.
\begin{lemma} \label{lem:almost-monochrom-rects}
	Let $\mm \in \{0,1,\dots, \Delta\}^{A \times B}$ be a bounded integral matrix of rank $r$.
	Then $\mm$ contains a $(1-\frac{1}{16r})$-monochromatic rectangle of size at least
	$|A||B| \cdot \exp(-\Delta^{O(\Delta)} \cdot \sqrt{r} \cdot \log{r})$.
\end{lemma}

To prove \cref{lem:almost-monochrom-rects}, we first define a distribution $\cD$ over the rectangles of $\mm$, and
then use the probabilistic method with respect to $\cD$ to show that there exists a large enough almost-monochromatic rectangle.
We begin with the technical ingredients:

\begin{definition}[Factorization norm]
	For a matrix $\mm \in \R^{A \times B}$, define its \emph{$\gamma_2$-norm} as
	\begin{align*}
		\gamma_2(\mm) \defeq \min \{ R \geq 0 &: \text{there are families of vectors $\{\vu_a\}_{a \in A}, \{\vv_b\}_{b \in B}$ so that} \\
			 &\; \mm_{a,b} = \inner{\vu_a, \vv_b} \text{ and } \norm{\vu_a}_2 \norm{\vv_b}_2 \leq R \text{ for all } a, b\}.
	\end{align*}
\end{definition}
In other words, $\gamma_2(\mm)$ gives the Euclidean length needed to factor the matrix $\mm$. 
Note that by rescaling the $\vu_a$'s and $\vv_b$'s, it is straightforward to guarantee $\norm{\vu_a}_2 \leq \sqrt{R}$ and $\norm{\vv_b}_2 \leq \sqrt{R}$ for all $a, b$ in the definition.
Here $\vu_a, \vu_b$ are vectors of any dimension (of course one may choose $\vu_a, \vu_b$ to have dimension $\rank(\mm)$).
The terminology $\gamma_2$-norm or factorization norm is indeed justified as $\gamma_2$ is a norm on the space of real matrices.

The following lemma is well-known:
\begin{lemma}[Lemma 4.2, \cite{linial2007complexity}] \label{lem:mm-gamma2-norm-bound}
	Any matrix $\mm \in \R^{A \times B}$ satisfies $\gamma_2(\mm) \leq \norm{\mm}_\infty \cdot \sqrt{\rank(\mm)}$.
\end{lemma}

It will be convenient for us to factor $\mm$ with vectors of the same Euclidean length, which comes at the expense of the dimension:
\begin{lemma}\label{lem:mm-factorization}
	For any matrix $\mm \in \R^{A \times B}$ and $s \geq \sqrt{\gamma_2(\mm)}$, there are vectors
	$\{\vu_a\}_{a \in A}, \{\vv_b\}_{b \in B}$ such that $\mm_{a,b} = \inner{\vu_a, \vv_b}$ and $\norm{\vu_a}_2 = \norm{\vv_b}_2 = s$ for all $a \in A$ and $b \in B$.
\end{lemma}
\begin{proof}
	Construct vectors $\vu_a, \vv_b$ each with 2-norm $\norm{\vu_a}_2, \norm{\vv_b}_2 \leq s$ and $\mm_{a,b} = \inner{\vu_a, \vv_b}$ using \cref{lem:mm-gamma2-norm-bound}.
	Then add $|A| + |B|$ new coordinates, where each $\vu_a$ and $\vv_b$ receives a ``private'' coordinate. Set the private coordinate of $\vu_a$ to $\sqrt{s^2 - \norm{\vu_a}_2^2}$, and similarly for $\vv_b$.
\end{proof}

We denote $N^n(0,1)$ as the $n$-dimensional standard Gaussian. 
Let $S^{n-1} \defeq \{ x \in \R^n : \norm{x}_2 = 1\}$ be the $n$-dimensional unit sphere.
The following argument is usually called hyperplane rounding in the context of approximation algorithms:
\begin{lemma}[Sheppard's formula] \label{lem:sheppards}
	Any vectors $\vu, \vv \in S^{n-1}$ with $\langle \vu, \vv \rangle = \alpha$ satisfy
	\[
		\Pr_{\vg \sim N^n(0,1)} \left[ \inner{\vg,\vu} \geq 0 \text{ and } \inner{\vg,\vv} \geq 0 \right] = h(\alpha) \defeq \frac{1}{2} \left( 1 - \frac{\arccos(\alpha)}{\pi} \right).
	\]
	More generally, any vectors $\vu, \vv \in \R^{n} \setminus \{\vzero\}$ satisfy
	\[
	\Pr_{\vg \sim N^n(0,1)} \left[ \inner{\vg, \vu} \geq 0 \text{ and } \inner{\vg,\vv} \geq 0 \right] = h\left(\frac{\inner{\vu,\vv}}{\norm{\vu}_2 \norm{\vv}_2}\right).
	\]
\end{lemma}

We can now define a suitable distribution over the rectangles of $\mm$ using the above tools. In particular, 
we want the probability that a rectangle contains an entry to be a function of the entry value.

\begin{lemma} \label{lem:rect-distribution}
	Let $\mm \in \{0,1, \dots, \Delta\}^{A \times B}$ be a bounded integral matrix of rank $r$. 
	Then for any $k$, there is a distribution $\mathcal{D}_k$ over the rectangles of $\mm$ so that for all $a \in A$ and $b \in B$,
	\[
		\Pr_{\mr \sim \mathcal{D}_k} \left[ (a,b) \in \mr \right] = \left(h \left( \frac{\mm_{a,b}}{\Delta \sqrt{r}} \right) \right)^k,
	\]
	where $h$ is the function defined in \cref{lem:sheppards}.
\end{lemma}
\begin{proof}
	We use \cref{lem:mm-factorization} to factor $\mm$, which gives vectors $\vu_a, \vv_b$ for all $a \in A$ and $b \in B$, such that $\mm_{a,b} = \inner{\vu_a, \vv_b}$, and $\norm{\vu_a}_2 = \norm{\vv_b}_2 = \Delta^{1/2} r^{1/4}$.
	Then we sample independent Gaussians $\vg_1, \dots, \vg_k \sim N^n(0,1)$ and set
	\[
		\mr_i \defeq \{ a \in A : \inner{\vu_a, \vg_i} \geq 0 \} \times \{ b \in B : \inner{\vv_b, \vg_i} \geq 0 \},
	\]
	and then $\mr \defeq \mr_1 \cap \dots \cap \mr_k$. Note that $\mr$ is indeed a rectangle.
	For each $i$ and each $(a,b) \in A \times B$ we have
	\[
		\Pr [ (a,b) \in \mr_i ] = h\left( \frac{\inner{\vu_a, \vv_b}}{\norm{\vu_a}_2 \norm{\vv_b}_2} \right) = h \left( \frac{\mm_{a,b}}{\Delta \sqrt{r}}\right).
	\]
	The overall expression follows by independence of the $k$ rectangles.
\end{proof}

Finally, we use the above distribution to show the existence of large almost-monochromatic rectangles.

\begin{proof}[of \cref{lem:almost-monochrom-rects}]
	Let $\cE_0 \;\dot\cup \cdots \dot\cup \; \cE_{\Delta}$ be the partition of the entries $A \times B$ based on entry values, so that $\cE_j \defeq \{(a,b) \in A \times B : \mm_{a,b} = j\}$.
	Let $m_j \defeq (64 r \Delta)^{(8\Delta)^j}$ for each $j = 0, \dots, \Delta$,
	and let $i$ be the index such that $m_i \cdot |\cE_i|$ is maximized.
	From this, we also get
	\begin{equation} \label{eq:Ei-bound}
		|\mm| = |A| \cdot |B| = \sum_{j=0}^\Delta |\cE_j| \leq \sum_{j=0}^\Delta \frac{m_i}{m_j} |\cE_i| \leq m_i |\cE_i|. 
	\end{equation}
	
	For notational convenience, 
	recall $h(\alpha) \defeq  \frac12 \left( 1 - \frac{\arccos(\alpha)}{\pi} \right)$, and let $c(j) \defeq h \left( \frac{j}{\Delta \sqrt{r}} \right)$.
	We observe that on $[0,1]$, the function $h$ is convex, monotone increasing, lowerbounded by $h(0) = 1/4$, upperbounded by $h(1) = 1/2$, and $h'(\alpha) = \frac{1}{2\pi \sqrt{1-\alpha^2}} \geq \frac{1}{2\pi}$.
	Additionally, the following claim about $c$ will be useful for our calculations later:
	\begin{claim} \label{lem:c-bounds}
		$\frac{c(j)}{c(j-1)} \geq 1 + \frac{4}{3 \pi \Delta \sqrt{r}}$ for $1 \leq j \leq \Delta$.
		Also, 
		$\frac{c(\Delta)}{c(0)} \leq 1 + \frac{4}{\pi \sqrt{r}}.$
	\end{claim}
	\begin{proof}
		We may assume $r \geq 2$. We use first order approximations for $c$. 
		For the first inequality, we also use the fact that $c(j-1) \leq 3/8$:
		\begin{align*}
			\frac{c(j)}{c(j-1)} &\geq \frac{c(j-1) + c'(j-1)}{c(j-1)} \geq 1 + \frac{1}{2 \pi \Delta \sqrt{r} \cdot c(j-1)} \geq 1 + \frac{4}{3\pi \Delta \sqrt{r}}.
		\end{align*}
		For the second inequality, we have
		\begin{align*}
			\frac{c(\Delta)}{c(0)} \leq \frac{c (0) + \Delta \cdot c'(\Delta) }{c(0)} = 
			1 +  \frac{4 \Delta}{\Delta \sqrt{r} \cdot 2\pi \sqrt{1 - 1/r}} \leq 1 + \frac{4}{\pi \sqrt{r}}.
		\end{align*}
	\end{proof}
	
	Next, let $\mathcal{D}_k$ be the distribution from \cref{lem:rect-distribution}, and generate $\mr \sim \mathcal{D}_k$
	for some choice of $k$ to be determined.
	We will show that there exists a $k$ such that $\mr$ is expected to be $(1-\frac{1}{16r})$-monochromatic with color $i$, and is sufficiently large.
	Specifically, the number of $i$-entries in $\mr$ is greater than the number all other entries in $\mr$ by a factor of $16r$ in expectation, and moreover, this difference is sufficiently large, which in turn means $\mr$ is sufficiently large.
	\begin{align*}
		&\phantom{{}={}} \E_{\mr \sim \mathcal{D}_k} \left[ |\cE_i \cap \mr| - 16r \sum_{j \neq i} |\cE_j \cap \mr| \right] \\
		&= | \cE_i | \cdot c(i)^k - 16r \sum_{j \neq i} c(j)^k |\cE_j| \tag{by \cref{lem:rect-distribution}}\\
		&\geq  | \cE_i | \cdot c(i)^k \left( 1 -  16r \sum_{j \neq i} \frac{c(j)^k m_i }{c(i)^k m_j} \right)
		\intertext{
			Suppose $\frac{c(j)^k m_i}{ c(i)^k m_j} \leq \frac{1}{64r \Delta}$ for each $j \neq i$, then we can conclude}
		&\geq | \cE_i | \cdot 4^{-k} (1 - \frac14) \tag{Since $c(i) \geq \frac14$}\\
		&\geq \frac{|A| |B|}{m_i} \frac{1}{2 \cdot 4^k}. \tag{by \cref{eq:Ei-bound}}
	\end{align*}

	\begin{claim}
		There exists a choice of $k$ such that $\frac{c(j)^k m_i}{ c(i)^k m_j} \leq \frac{1}{64r \Delta}$ for all $j \neq i$.
	\end{claim}
	\begin{proof}
		We consider two cases:
		\begin{enumerate}[(1)]
			\item $0 \leq j < i$: In this case we can bound
			$\frac{c(j)^k m_i}{ c(i)^k m_j} \leq \left(\frac{c(i-1)}{c(i)} \right)^{k} \frac{m_i}{m_0}$.
			To get our claim, it suffices to choose $k$ to satisfy
			\begin{align*}
				\left(\frac{c(i-1)}{c(i)} \right)^{k} \frac{m_i}{m_0} &\leq \frac{1}{64r\Delta}  \\
				\Leftarrow \hspace{6em} k  &\geq \frac{((8 \Delta)^i+1) \log (64 r \Delta)}{\log \frac{c(i)}{c(i-1)}}.
			\end{align*} 
			Using the lower bound from \cref{lem:c-bounds}, along with $\log(1+x) \geq x/2$ for $x \leq 1$, we conclude it suffices to choose $k$ to satisfy
			\begin{equation}\label{eq:k-lowerbound}
			k \geq ((8 \Delta)^i+1) \log (64 r \Delta) \cdot \frac{3}{2} \pi \Delta \sqrt{r}.
			\end{equation}
			\item $i < j \leq \Delta$: In this case we can bound
			$\frac{c(j)^k m_i}{ c(i)^k m_j} \leq \left(\frac{c(\Delta)}{c(0)} \right)^{k} \frac{m_i}{m_{i+1}}$.
			To get our claim, it suffices to choose $k$ to satisfy
			\begin{align*}
				\left(\frac{c(\Delta)}{c(0)} \right)^{k} \frac{m_i}{m_{i+1}} &\leq \frac{1}{64r \Delta} \\
				\Leftarrow \hspace{6em} k &\leq \frac{((8 \Delta)^{i+1} - (8\Delta)^i -1) \log (64r \Delta)}{ \log \frac{c(\Delta)}{c(0)}}.
			\end{align*}
			Using the upper bound from \cref{lem:c-bounds}, we conclude it suffices to choose $k$ to satisfy
			\begin{equation} \label{eq:k-upperbound}
			k \leq ((8\Delta)^{i+1} - (8\Delta)^{i}-1) \log (64 r \Delta) \cdot \frac{\pi \sqrt{r}}{4}.
			\end{equation}
		\end{enumerate}
		To choose $k$ to simultaneously satisfy the two cases, we first verify that the lower and upper bound for $k$ in \cref{eq:k-lowerbound} and \cref{eq:k-upperbound} are consistent. 
		Indeed when $i \neq 0$, we have
		\[
			\frac{3}{2} ((8 \Delta)^i+1) \Delta \leq \frac14 ((8\Delta)^{i+1} - (8\Delta)^{i}-1),
		\]
		so we may choose $k$ to be equal to the lower bound.
		If $i = 0$, then the lower bound from \cref{eq:k-lowerbound} does not apply, so we choose $k$ to be equal to the upper bound. 
	\end{proof}
	For any established choice of $k$, we always have $k \leq (8 \Delta)^{\Delta} \log(64 r \Delta) \pi \Delta \sqrt{r}$.
	Moreover, we have $\log m_i \leq (8 \Delta)^\Delta \log(64 r \Delta)$.
        We conclude that
	\begin{align*}
                                                             \E_{\mr \sim \mathcal{D}_k} \left[ |\cE_i \cap \mr| - 16r \sum_{j \neq i} |\cE_j \cap \mr| \right] 
	&\geq \frac{|A| |B|}{m_i} \frac{1}{2 \cdot 4^k} \\
	&\geq \exp(-\Delta^{O(\Delta)}\cdot  \sqrt{r} \cdot \log r) |A| |B|.
	\end{align*}
        Then any $\mr$ attaining this expectation will simultaneously satify
        $|\mr| \geq \exp(-\Delta^{O(\Delta)}\cdot  \sqrt{r} \cdot \log r) |A| |B|$ and be
        $(1-\frac{1}{16r})$-monochromatic.
\end{proof}
\section{Proof of \cref{thm:main}}

By \cref{lem:almost-monochrom-rects}, 
we know any submatrix $\ms$ of $\mm$ contains a $(1-\frac{1}{16r})$-monochromatic rectangle of size
$\exp(-\Delta^{O(\Delta)} \cdot \sqrt{r} \cdot \log r) |\ms|$.
Therefore, by \cref{lem:almost-monochrom-rec-reduction}, $\ms$ contains a monochromatic rectangle of size
$\frac{1}{8} \exp(-\Delta^{O(\Delta)} \cdot \sqrt{r} \cdot \log r) |\ms|$.
We substitute $\delta(r) = \Delta^{O(\Delta)} \sqrt{r} \log r$ in \cref{lem:comm-protocol} to get
\begin{align*}
	CC^{\det}(\mm) &\leq \Theta \left(\log \Delta +  \log^2 r  + \sum_{i=0}^{\log r} \delta(r/2^i) \right).
\end{align*}
The summation simplifies as follows:
\begin{align*}
	\sum_{i=0}^{\log r} \delta(r/2^i) &=  \sum_{i=0}^{\log r} \Delta^{O(\Delta)} \sqrt{r} \cdot 2^{-i/2} \cdot \log\Big(\frac{r}{2^i}\Big) \\
	&\leq \Delta^{O(\Delta)} \sqrt{r} \log r \cdot \sum_{i=0}^{\log r} 2^{-i/2},
\end{align*}
where the sum converges. We ignore the lower order terms in the communication complexity expression to conclude
\[
	CC^{\det}(\mm) \leq \Theta \left(\Delta^{O(\Delta)} \sqrt{r} \log r  \right).
\]
\qed

\section{Acknowledgment}
We thank Sam Fiorini for helpful discussions and anonymous reviewers for helpful feedback.

\bibliographystyle{alpha}
\bibliography{main}

\newpage

\appendix

\end{document}